\documentclass[letterpaper, 10 pt, conference]{ieeeconf}  %

\IEEEoverridecommandlockouts                              %
\usepackage[noadjust,nospace]{cite}

\usepackage{amsmath,amssymb}       %
\usepackage{nicefrac}
\usepackage{mleftright}
\usepackage{cases}
\usepackage{empheq}
\usepackage{stmaryrd}
\usepackage{bm}
\usepackage{bbm}
\DeclareSymbolFontAlphabet{\amsmathbb}{AMSb}
\usepackage{mathtools}
\usepackage{accents}
\usepackage[colorinlistoftodos]{todonotes}

\usepackage{siunitx}
\usepackage{caption}
\usepackage{subcaption}

\graphicspath{./Figures/}

\newcommand\munderbar[1]{%
  \underaccent{\bar}{#1}}

\usepackage{xspace}
\usepackage{microtype}

\usepackage{hyperref}
\usepackage{cleveref}

\crefname{equation}{Eq.}{Eqs.}
\crefname{pluralequation}{Eqs.}{Eqs.}

\crefname{algorithm}{Algorithm}{Algorithm}

\crefname{figure}{Fig.}{Figs.}
\crefname{pluralfigure}{Figs.}{Figs.}

\crefname{section}{Sect.}{Sects.}
\crefname{pluralsection}{Sects.}{Sects.}

\crefname{table}{Table}{Table}
\crefname{pluraltable}{Tables}{Tables}

\crefname{definition}{Def.}{Def.}
\crefname{pluraldefinition}{Defs.}{Defs.}

\crefname{theorem}{Theorem}{Theorems}
\crefname{pluraltheorem}{Theorems}{Theorems}

\crefname{lemma}{Lemma}{Lemmas}
\crefname{plurallemma}{Lemmas}{Lemmas}

\crefname{example}{Example}{Example}
\crefname{pluralexample}{Examples}{Examples}

\crefname{problem}{Problem}{Problem}
\crefname{pluralproblem}{Problems}{Problems}

\crefname{assumption}{Assumption}{Assumption}
\crefname{pluralassumption}{Assumptions}{Assumptions}

\crefname{remark}{Remark}{Remark}
\crefname{pluralremark}{Remarks}{Remarks}

\crefname{proposition}{Proposition}{Proposition}
\crefname{pluralproposition}{Propositions}{Propositions}

\crefname{appendix}{Appendix}{Appendices}
\crefname{pluralappendix}{Appendices}{Appendices}

\usepackage{tikz} 
\usepackage{pgfplots}
\usepgfplotslibrary{external,fillbetween,colorbrewer}
\pgfplotsset{compat=newest}

\DeclareUnicodeCharacter{2212}{−}
\usepgfplotslibrary{external,}
\usetikzlibrary{patterns,shapes.arrows}

\definecolor{red}{rgb}{0.745,0.192,0.102}
\definecolor{darkgreen}{RGB}{34,161,55}
\definecolor{ruhuisstijlrood}{rgb}{0.745,0.192,0.102}
\definecolor{ruhuisstijlzwart}{rgb}{0,0,0}
\definecolor{ruhuisstijlwit}{rgb}{0.98,0.98,0.98}

\definecolor{plotblue}{rgb}{0.1,0.498039215686275,0.9549019607843137}

\newcolumntype{H}{>{\setbox0=\hbox\bgroup}c<{\egroup}@{}}
\usepackage{colortbl}
\definecolor{LightCyan}{rgb}{0.88,1,1}

\definecolor{red}{rgb}{0.745,0.192,0.102}
\definecolor{darkgreen}{RGB}{34,161,55}
\definecolor{ruhuisstijlrood}{rgb}{0.745,0.192,0.102}
\definecolor{ruhuisstijlzwart}{rgb}{0,0,0}
\definecolor{ruhuisstijlwit}{rgb}{0.98,0.98,0.98}

\usepackage{mdframed}

\usetikzlibrary{patterns,patterns.meta,arrows,arrows.meta,calc,shapes,shadows,decorations.pathmorphing,decorations.pathreplacing,automata,shapes.multipart,positioning,shapes.geometric,fit,circuits,trees,shapes.gates.logic.US,fit, shadows.blur,shapes.symbols,intersections}

\usepackage{fontawesome5}

\usepackage{tabularx} %

\usepackage{tabularx}
\usepackage{booktabs}
\usepackage{scrextend}
\usepackage{threeparttable}

\newtheorem{example}{Example}
\newtheorem{theorem}{Theorem}

\newtheorem{definition}{Definition}

\newtheorem{lemma}{Lemma}

\newtheorem{remark}{Remark}

\newcommand*{\policy}{\ensuremath{\mu}}

\newcommand*{\States}{\ensuremath{\mathbb{X}}}
\newcommand*{\Actions}{\ensuremath{\mathbb{U}}}
\newcommand*{\Disturb}{\ensuremath{\mathbb{V}}}

\newcommand*{\Labels}{\ensuremath{\mathbb{L}}}
\newcommand*{\initState}{\ensuremath{\bar{x}}}
\newcommand*{\transfunc}{\ensuremath{\mathbb{T}}}
\newcommand*{\transfuncR}{\ensuremath{\mathbb{T}^R}}

\newcommand*{\labelfunc}{\ensuremath{h}}

\newcommand*{\mdp}{\ensuremath{\mathcal{D}}}

\newcommand*{\rmdp}{{\ensuremath{\mathcal{M}}}}

\newcommand*{\MDP}{\ensuremath{\tuple{\States,\initState,\Actions,\transfunc,\Labels,\labelfunc}}}
\newcommand*{\MDPi}[1]{\ensuremath{\tuple{\States_{#1},{\initState}_{#1},\Actions_{#1},\transfunc_{#1},\Labels,\labelfunc_{#1}}}}
\newcommand*{\RMDP}{\ensuremath{\tuple{\States,\initState,\Actions,\Disturb,\transfuncR,\Labels,\labelfunc}}}
\newcommand*{\RMDPi}[1]{\ensuremath{\tuple{\States_{#1},{\initState}_{#1},\Actions_{#1},\Disturb_{#1},\transfuncR_{#1},\Labels,\labelfunc_{#1}}}}

\newcommand*{\adversary}{\ensuremath{\tau}}

\newcommand*{\NN}{\mathbb{N}}

\newcommand*{\RR}{\mathbb{R}}

\newcommand*{\tuple}[1]{\left(#1\right)}

\newcommand*{\support}[1]{\textsf{support}(#1)}

\newcommand*{\inv}{\ensuremath{{-1}}}

\newcommand*{\Relation}{\mathcal{R}}
\newcommand*{\liftRelation}{\mathcal{R}^\mathcal{P}}
\newcommand*{\Relinv}[1]{\Relation^\inv({#1})}
\newcommand*{\Rel}[1]{\Relation({#1})}

\newcommand*{\Prob}{\mathbb{P}}

\newcommand*{\distr}[1]{\mathcal{P}(#1)}
\newcommand*{\Borel}[1]{\ensuremath{\mathcal{B}(#1)}}

\newcommand*{\Gauss}[2]{\ensuremath{\mathcal{N}(#1 , #2)}}

\DeclareMathOperator*{\argmin}{argmin}

\usepackage[acronym]{glossaries}

\glsdisablehyper

\newacronym[]{LTL}{LTL}{linear temporal logic}
\newacronym[]{iid}{i.i.d.}{independent and identically distributed}
\newacronym[]{UAV}{UAV}{unmanned aerial vehicle}

\newacronym[]{PAC}{PAC}{probably approximately correct}
\newacronym[]{DRO}{DRO}{distributionally robust optimization}
\newacronym[plural=MCs,firstplural=Markov chains (MCs)]{MC}{MC}{Markov chain}
\newacronym[plural=DTMCs,firstplural=discrete-time Markov chains (DTMCs)]{DTMC}{DTMC}{discrete-time Markov chain}
\newacronym[plural=MDPs,firstplural=Markov decision processes (MDPs)]{MDP}{MDP}{Markov decision process}
\newacronym[plural=pMCs,firstplural=parametric MCs (pMCs)]{pMC}{pMC}{parametric MC}
\newacronym[plural=rMCs,firstplural=robust MCs (rMCs)]{rMC}{rMC}{robust MC}
\newacronym[plural=prMCs,firstplural=parametric robust MCs (prMCs)]{prMC}{prMC}{parametric robust MC}
\newacronym[plural=prMDPs,firstplural=parametric robust MDPs (prMDPs)]{prMDP}{prMDP}{parametric robust MDP}
\newacronym[plural=iMDPs,firstplural=interval Markov decision processes (iMDPs)]{iMDP}{iMDP}{interval Markov decision process}
\newacronym[plural=aMDPs,firstplural=augmented MDPs (aMDPs)]{aMDP}{aMDP}{augmented MDP}
\newacronym[plural=POMDPs,firstplural=partially observable Markov decision processes (POMDPs)]{POMDP}{POMDP}{partially observable Markov decision process}

\definecolor{color1}{RGB}{55,126,184} %
\definecolor{color2}{RGB}{228,26,28} %
\definecolor{color3}{RGB}{77,175,74} %
\definecolor{color4}{RGB}{152,78,163} %
\definecolor{color5}{RGB}{255,127,0} %
\definecolor{color6}{rgb}{0.5, 1.0, 0.83} %
\definecolor{color7}{rgb}{1.0, 0.0, 1.0} %
\definecolor{color8}{rgb}{0.66, 0.66, 0.66} %

\newcommand{\scatterplotstorm}[6]{%
	\begin{tikzpicture}
	\begin{axis}[
	width=\scatterplotsize,
	height=\scatterplotsize,
	axis equal image,
	xmin=0.01,
	ymin=0.01,
	ymax=22000,
	xmax=22000,
	xmode=log,
	ymode=log,
	axis x line=bottom,
	axis y line=left,
	xtick={0.01,0.1,1,5,20,100,1000,3000},
	xticklabels={0.01,0.1,1,5,20,100,1000,3000},
	extra x ticks = {10000},
	extra x tick labels = {Timeout},
	extra x tick style = {grid = major},
	ytick={0.01,0.1,1,5,20,100,1000,3000},
	yticklabels={0.01,0.1,1,5,20,100,1000,3000},
	extra y ticks = {10000},
	extra y tick labels = {Timeout},
	extra y tick style = {grid = major},
	xlabel={#3},
	xlabel style={font=\small,yshift=18pt,xshift=-14pt},
	ylabel={#5},
	ylabel style={font=\small,yshift=-0.55cm},
	yticklabel style={font=\tiny},
	xticklabel style={rotate=290,anchor=west,font=\tiny},
	legend pos=north west,
	legend columns=-1,
	legend style={at={(0.4,0.15)},nodes={scale=0.75, transform shape},inner sep=1.5pt},
        clip mode=individual,
	]
	
	\addplot[
	scatter,
	only marks,
	scatter/classes={
		pMC={mark=*,color1,mark size=1.5},
		prMC={mark=triangle*,color2,mark size=1.75}
	},
	scatter src=explicit symbolic
	]%
	table [col sep=semicolon,x=#2,y=#4,meta=Type] {#1};
	\ifthenelse{\NOT\equal{#6}{false}}{\legend{pMC, prMC}}{}
	\addplot[no marks] coordinates {(0.001,0.001) (10000,10000) };
	\addplot[no marks, densely dotted] coordinates {(0.001,0.01) (1000,10000)};
	\addplot[no marks, densely dotted] coordinates {(0.01,0.001) (10000,1000)};

        \draw [latex-] (axis cs:50,500)-- +(-6pt,5pt) node[left, xshift=5pt, yshift=5pt] {$10\times$ faster};
 
	\end{axis}
	\end{tikzpicture}
}

\newcommand{\scatterplotstormB}[6]{%
	\begin{tikzpicture}
	\begin{axis}[
	width=\scatterplotsize,
	height=\scatterplotsize,
	axis equal image,
	xmin=0.001,
	ymin=0.001,
	ymax=22000,
	xmax=22000,
	xmode=log,
	ymode=log,
	axis x line=bottom,
	axis y line=left,
	xtick={0.01,0.1,1,5,20,100,1000,3000},
	xticklabels={0.01,0.1,1,5,20,100,1000,3000},
	extra x ticks = {10000},
	extra x tick labels = {Timeout},
	extra x tick style = {grid = major},
	ytick={0.01,0.1,1,5,20,100,1000,3000},
	yticklabels={0.01,0.1,1,5,20,100,1000,3000},
	extra y ticks = {10000},
	extra y tick labels = {Timeout},
	extra y tick style = {grid = major},
	xlabel={#3},
	xlabel style={font=\small,yshift=18pt,xshift=-5pt},
	ylabel={#5},
	ylabel style={font=\small,yshift=-0.55cm,xshift=-0.2cm},
	yticklabel style={font=\tiny},
	xticklabel style={rotate=290,anchor=west,font=\tiny},
	legend pos=north west,
	legend columns=-1,
	legend style={at={(0.4,0.15)},nodes={scale=0.75, transform shape},inner sep=1.5pt},
        clip mode=individual,
	]
	
	\addplot[
	scatter,
	only marks,
	scatter/classes={
		pMC={mark=*,color1,mark size=1.5},
		prMC={mark=triangle*,color2,mark size=1.75}
	},
	scatter src=explicit symbolic
	]%
	table [col sep=semicolon,x=#2,y=#4,meta=Type] {#1};
	\ifthenelse{\NOT\equal{#6}{false}}{\legend{pMC, prMC}}{}
	\addplot[no marks] coordinates {(0.001,0.001) (10000,10000) };
	\addplot[no marks, densely dotted] coordinates {(0.001,0.01) (1000,10000)};
	\addplot[no marks, densely dotted] coordinates {(0.01,0.001) (10000,1000)};

        \draw [latex-] (axis cs:100,1000)-- +(-6pt,8pt) node[left] {$10\times$ faster};
 
	\end{axis}
	\end{tikzpicture}
}

\usetikzlibrary{external}
\newif\iftikzexternal
\global\tikzexternaltrue

\title{\LARGE \bf
Probabilistic Alternating Simulations for Policy Synthesis in \\ Uncertain Stochastic Dynamical Systems
}

\author{Thom Badings and Alessandro Abate%
\thanks{This research was supported by EPSRC grant EP/Y028872/1, Mathematical Foundations of Intelligence: An ``Erlangen Programme'' for AI.}%
\thanks{Thom Badings and Alessandro Abate are with the Department of
Computer Science, University of Oxford, United Kingdom. {\tt\small \{thom.badings,alessandro.abate\}@cs.ox.ac.uk}.}%
}

\begin{document}

\maketitle
\thispagestyle{empty}
\pagestyle{empty}

\begin{abstract}
A classical approach to formal policy synthesis in stochastic dynamical systems is to construct a finite-state abstraction, often represented as a Markov decision process (MDP). The correctness of these approaches hinges on a behavioural relation between the dynamical system and its abstraction, such as a probabilistic simulation relation. However, probabilistic simulation relations do not suffice when the system dynamics are, next to being stochastic, also subject to nondeterministic (i.e., set-valued) disturbances. In this work, we extend probabilistic simulation relations to systems with both stochastic and nondeterministic disturbances. Our relation, which is inspired by a notion of alternating simulation, generalises existing relations used for verification and policy synthesis used in several works. Intuitively, our relation allows reasoning probabilistically over stochastic uncertainty, while reasoning robustly (i.e., adversarially) over nondeterministic disturbances. We experimentally demonstrate the applicability of our relations for policy synthesis in a 4D-state Dubins vehicle.
\end{abstract}

\section{Introduction}
The synthesis of (control) policies for dynamical systems that provably satisfy specific requirements is crucial for their deployment in safety-critical scenarios. 
We consider systems modelled as \emph{Markov decision processes} (MDPs)~\cite{DBLP:books/wi/Puterman94} with continuous state and action spaces.
These (continuous) MDPs capture nonlinear and stochastic dynamics and are thus widely applicable for modelling systems in uncertain environments.
Classical objectives in automatic control, such as stabilisation and tracking, are insufficient to capture the complex objectives needed for many systems.
Instead, we consider objectives in temporal logic, such as linear temporal logic (LTL) %
and probabilistic computation tree logic (PCTL). %
Temporal logic enables formulating complex, high-level objectives involving periodic, sequential, or reactive tasks~\cite{BaierKatoen08}.

Approaches to policy synthesis with temporal logic specifications broadly fall into two categories.
First, certificate-based approaches aim to find a (Lyapunov-like) function that implies the satisfaction of a specification~\cite{DBLP:conf/cav/AbateGR25,DBLP:journals/tac/PrajnaJP07}.
The second category, which we focus on in this paper, replaces the continuous MDP (the ``\emph{concrete}'' system) with a simpler, finite MDP (the ``\emph{abstraction}'') and uses model checking techniques~\cite{BaierKatoen08} to compute a policy on this abstraction.
These abstractions are classically model-based~\cite{DBLP:journals/automatica/AbatePLS08,DBLP:journals/tac/ZamaniEMAL14,DBLP:journals/tac/LahijanianAB15,DBLP:journals/jair/BadingsRAPPSJ23,DBLP:conf/hybrid/MathiesenHL25}, but recent works study data-driven approaches as well~\cite{DBLP:conf/l4dc/GraciaBLL24,pmlr-v283-nazeri25a,DBLP:journals/csysl/LavaeiSFZ23}.
Although abstractions tend to explode with the state dimension, they can handle rich specifications and natively capture stochasticity~\cite{DBLP:journals/automatica/LavaeiSAZ22}.

The correctness of abstraction techniques hinges on a \emph{behavioural relation} between the concrete system and the abstraction~\cite{DBLP:books/daglib/0032856}.
Such a relation ensures that any policy for the abstraction can be \emph{refined} into a policy for the concrete system \emph{with equivalent performance guarantees}.
Relations for synthesis in \emph{nonstochastic} systems have been well-studied, leading to, e.g., (approximate) simulation relations~\cite{DBLP:journals/tac/GirardP07}, feedback refinement relations~\cite{DBLP:journals/tac/ReissigWR17}, and memoryless concretisation relations~\cite{DBLP:conf/hybrid/CalbertMGJ24}.
For \emph{stochastic} systems such as MDPs, most papers leverage (approximate) \emph{probabilistic simulations} to ensure the soundness of abstraction techniques~\cite{DBLP:journals/siamco/HaesaertSA17,DBLP:journals/jair/BadingsRAPPSJ23}.

Loosely speaking, the system (I) probabilistically simulates another system (II) if, for every policy of system (II), there exists a policy for system (I) such that their output behaviour is equivalent.
Technically, probabilistic simulation thus requires that the closed-loop system under a given policy is a stochastic process.
As a result, probabilistic simulation does not allow for nondeterministic (i.e., set-valued) disturbances in the MDP's dynamics.
Such disturbances naturally arise in systems with uncertain parameters and multi-agent systems~\cite{BaierKatoen08}.

To solve this problem, we extend probabilistic simulation relations to systems with both stochastic and nondeterministic dynamics.
We model such systems as \emph{robust MDPs} (RMDPs), which extend MDPs with \emph{sets of probability distributions}, often as convex polytopes for tractability~\cite{DBLP:journals/ior/NilimG05,DBLP:journals/mor/WiesemannKR13}.
While RMDPs with finite state and action spaces (and with finite state but continuous action spaces~\cite{DBLP:conf/hybrid/Delimpaltadakis23}) are well-studied~\cite{DBLP:conf/aistats/PanagantiK22}, we consider their full generalisation to continuous spaces.
We develop a behavioural relation for continuous RMDPs, inspired by the notion of \emph{alternating simulation} for two-player (stochastic) games~\cite{DBLP:conf/concur/AlurHKV98,DBLP:conf/sofsem/ZhangP12}.
Much like~\cite{DBLP:journals/siamco/HaesaertSA17} extends probabilistic simulation to continuous MDPs, we extend probabilistic \emph{alternating} simulation to continuous MDPs with \emph{set-valued dynamics}.
Our relation treats the nondeterminism as a second player in a game, which allows robust reasoning against these disturbances.
Thus, probabilistic alternating simulations allow reasoning probabilistically over stochasticity, and robustly over nondeterministic disturbances.
Our relations are closely related to~\cite{DBLP:journals/automatica/ZhongLZC23}, which follows a slightly different formalisation not based on \emph{alternating} notions of simulation.

In summary, our main contribution is a novel probabilistic alternating simulation for stochastic dynamical systems with uncertain dynamics.
After the preliminaries in~\cref{sec:preliminaries}, we present our theoretical results in~\cref{sec:RMDPs,sec:PASR}.
We also discuss how our relations generalise some others used in existing works.
To showcase the applicability, we use our results in \cref{sec:experiments} to synthesise policies with reach-avoid guarantees for a Dubins vehicle with a 4D state space.

\section{Preliminaries}
\label{sec:preliminaries}

A Polish space is a separable completely metrisable topo\-logical space. %
The power set %
over $X$ is written $2^X$.
A probability space $\tuple{\Omega, \mathcal{F}, \Prob}$ consists of a sample space $\Omega$, a $\sigma$-algebra $\mathcal{F}$, and a probability measure $\Prob \colon \mathcal{F} \to [0,1]$.
We denote the Borel $\sigma$-algebra over a set $X$ by $\Borel{X}$.
The set of all distributions over an (in)finite set $X$ is denoted by $\distr{X}$.
A set $\Relation \subseteq X \times Y$ is called a \emph{binary relation} between sets $X$ and $Y$, for which we write $\Rel{x} \coloneqq \{ y \in Y: (x,y) \in \Relation \}$ and $\Relinv{y} \coloneqq \{ x \in X : (x,y) \in \Relation \}$.
For subsets $X' \subset X$ and $Y' \subset Y$, we write $\Rel{X'} \coloneqq \{ y \in Y: \exists x \in X', (x,y) \in \Relation \}$ and $\Relinv{Y'} \coloneqq \{ x \in X: \exists y \in Y', (x,y) \in \Relation \}$.
The relation $\Relation$ is \emph{single-valued} if $|\Rel{x}| = 1$ for all $x \in X$, in which case $\Relation$ induces a partition into equivalence classes.

\subsection{Continuous Markov decision processes}
We consider discrete-time nonlinear stochastic systems, modelled as a (continuous) Markov decision process (MDP).

\begin{definition}[MDP]
    \label{def:MDP}
    A (continuous) Markov decision process (MDP) is a tuple $\mdp = \MDP$, where
    \begin{itemize}
        \item $\States$ is a Polish space, called the \emph{state space},
        \item $\initState \in \distr{\States}$ is a probability measure on $(\States,\Borel{\States})$ modelling the \emph{initial state distribution},
        \item $\Actions$ is a Polish space, called the \emph{action space},
        \item $\transfunc$ is a \emph{stochastic kernel} that assigns to each $x \in \States$ and $u \in \Actions$ a probability measure $\transfunc(\cdot \mid x,u)$ over $(\States,\Borel{\States})$,
        \item $\Labels$ is a finite set of labels, and
        \item $\labelfunc \colon \Borel{\States} \to 2^\Labels$ is a measurable \emph{labelling function} that assigns to each state a (possibly empty) subset of labels.
    \end{itemize}
\end{definition}

\begin{example}
    \label{example}
    Consider a Dubins vehicle with a 4D state $[x_k, y_k, \theta_k, V_k] \in \RR^4$, whose dynamics are defined as
    \begin{align}
        \label{eq:Dubins}
        \begin{bmatrix}
            x_{k+1} \\ y_{k+1} \\ \theta_{k+1} \\ V_{k+1}
        \end{bmatrix} = \begin{bmatrix}
            x_k \\ y_k \\ \theta_k \\ \beta \cdot V_k
        \end{bmatrix} + \delta \cdot \begin{bmatrix}
            V_k \cdot \cos{\theta_k} \\
            V_k \cdot \sin{\theta_k} \\
            \alpha \cdot u_k + w_k \\
            u'_k,
        \end{bmatrix},
    \end{align}
    with time discretization $\delta > 0$, steering sensitivity $\alpha > 0$, drag coefficient $\beta > 0$, and Gaussian noise $w_k \sim \Gauss{0}{0.1}$.
    We model this system as MDP $\mdp$ with states $\States = \RR^4$, inputs $[u_k, u'_k] \in \Actions \subset \RR^2$, and stochastic kernel $\transfunc$ given by \cref{eq:Dubins}.
\end{example}

\textbf{Policies.}
The actions in an MDP are selected by a Markov policy, which acts as a time-varying feedback controller.

\begin{definition}[Markov policy]
    \label{def:policy}
    A (Markov) policy $\policy$ for an MDP $\mdp = \MDP$ is a sequence $\policy = (\policy_0, \policy_1, \ldots)$, where each $\policy_k \colon \States \to \distr{\Actions}$ is a universally measurable map.
\end{definition}

Observe that the policy maps from \emph{states} $x \in \States$ (and not from \emph{labels} $y \in \Labels$, as with policies for partially observable MDPs).
Instead, the labelling function of the MDP defines the space in which we express the desired system behaviour.

\textbf{Execution.}
For a policy $\policy$, the sequence of states $x_0, x_1, \ldots$ is given by sampling $x_0 \sim \initState$ and $x_{k+1} \sim \transfunc(\cdot \mid x_k, \mu_k(x_k))$ for all $k \in \NN$.
Fixing a policy for an MDP thus creates a Markov process in the space of \emph{executions}.
Formally, this execution $\{ x_k \}_{k \in \NN}$ is a stochastic process defined on the probability space $(\Omega, \Borel{\Omega}, \Prob_{\mdp}^{\policy})$ with the sample space $\Omega = \States \times \States \times \cdots$ and the Borel $\sigma$-algebra $\Borel{\Omega}$ over $\Omega$, and where the probability measure $\Prob_{\mdp}^{\policy} \colon \Borel{\Omega} \to [0,1]$ is uniquely defined~\cite[Proposition~7.45]{Bertsekas.Shreve78}.
A \emph{sampled execution} is a sequence $(x_0, x_1, \ldots) \in \Omega$ of states such that $x_{k+1} \in \support{\transfunc(\cdot \mid x_k, \policy_k(x_k))} \,\forall k \in \NN$.
Executions over finite horizons are defined analogously.

\subsection{Probabilistic simulation relations}
We review the \emph{probabilistic simulation relation} (PSR) for MDPs proposed by~\cite{DBLP:journals/siamco/HaesaertSA17}.\footnote{We remark that~\cite{DBLP:journals/siamco/HaesaertSA17} considers so-called \emph{general MDPs}, which are a generalisation of our (continuous) MDPs with a metric on the output space. We instead restrict ourselves to the labelling function $\labelfunc$ in \cref{def:MDP}.}
A PSR is based on a binary relation {$\Relation \subseteq \States_1 \times \States_2$} between the states of two MDPs $\mdp_i = \MDPi{i}, \, i=1,2$ sharing the same set of labels $\Labels$.
Towards recapping this result, we define the \emph{lifting} of such a relation from states to distributions over states.

\begin{definition}[Lifted relation~\cite{DBLP:journals/siamco/HaesaertSA17}]
    \label{def:lifting}
    Let $\Relation \subseteq \States_1 \times \States_2$ be a relation between $(\States_1, \Borel{\States_1})$ and $(\States_2, \Borel{\States_2})$.
    The relation $\liftRelation \subseteq \distr{\States_1, \Borel{\States_1}} \times \distr{\States_2, \Borel{\States_2}}$ is called a \emph{lifting of relation $\Relation$} if $(\Delta, \Theta) \in \liftRelation$ holds for all $\Delta \in \distr{\States_1, \Borel{\States_1}}$ and $\Theta \in \distr{\States_2, \Borel{\States_2}}$ for which there exists a probability space $(\States_1 \times \States_2, \Borel{\States_1 \times \States_2}, \mathbb{W})$ satisfying:
    \begin{enumerate}
        \item for all $X_1 \in \Borel{\States_1}$ it holds that $\mathbb{W}(X_1, \States_2) = \Delta(X_1)$,
        \item for all $X_2 \in \Borel{\States_2}$ it holds that $\mathbb{W}(\States_1, X_2) = \Theta(X_2)$,
        \item $\mathbb{W}(\Relation) = 1$.
    \end{enumerate}
\end{definition}

Intuitively, two distributions $\Delta \in \distr{\States_1, \Borel{\States_1}}$ and $\Theta \in \distr{\States_2, \Borel{\States_2}}$ are related, i.e., $(\Delta, \Theta) \in \liftRelation$, if there exists another distribution in the product space $\States_1 \times \States_2$ such that the marginals recover $\Delta$ and $\Theta$, and such that the probability $\mathbb{W}(\Relation)$ of the event $\Relation$ is one.

\begin{example}
    \label{example:lifting}
    Consider the relation $\Relation \subseteq \RR \times \RR_{\geq 0}$ defined as $(x,y) \in R \iff |x| = y$, i.e., $(x,y)$ are related if the absolute value of $x$ equals $y$.
    Consider two uniform distributions $\Delta = U(-1,1)$ and $\Theta = U(0,1)$.
    These distributions are related by the lifting of $\Relation$, i.e., $(\Delta, \Theta) \in \liftRelation$, since the uniform distribution over the set $W$ depicted in \cref{fig:lifting} satisfies the three conditions in \cref{def:lifting}.
\end{example}

We now recap the probabilistic simulation relation from~\cite{DBLP:journals/siamco/HaesaertSA17} as a relation between two continuous MDPs.

\begin{definition}[Prob. simulation~\cite{DBLP:journals/siamco/HaesaertSA17}]
    \label{def:PSR}
    Consider two MDPs $\mdp_i = \MDPi{i}$, $i=1,2$ with the same set of labels $\Labels$.
    A single-valued\footnote{These results may be generalised beyond single-valued relations, which, however, requires a more involved policy refinement step.} binary relation $\Relation \subseteq \States_1 \times \States_2$ is a \emph{probabilistic simulation relation (PSR)} from $\mdp_2$ to $\mdp_1$ if:
    \begin{enumerate}
        \item %
        for the initial distributions, we have $(\initState_1, \initState_2) \in \liftRelation$;
        \item %
        for all $(x_1, x_2) \in \Relation$, we have%
        \begin{equation}
        \begin{split}
            \label{eq:PSR_next}
            &\forall u_2 \in \Actions_2, \, \exists u_1 \in \Actions_1 \, \text{such that}
            \\
            &\quad \big( \transfuncR_1(\cdot \mid x_1, u_1), \transfuncR_2(\cdot \mid x_2, u_2) \big) \in \liftRelation;
        \end{split}
        \end{equation}
        \item %
        for all $(x_1, x_2) \in \Relation$, we have $\labelfunc_1(x_1) = \labelfunc_2(x_2)$.
    \end{enumerate}
\end{definition}

\begin{figure}[t!]
\centering
\includegraphics[scale=0.82]{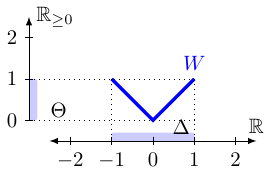}
\caption{Uniform distributions $\Delta = U(-1,1)$ and $\Theta = U(0,1)$ for the relation from \cref{example:lifting}. The uniform distribution over the set $W$ satisfies the conditions for a lifting in \cref{def:lifting}.}
\label{fig:lifting}
\end{figure}

These conditions state that: (1) the initial state distributions are related, (2) every pair of related states leads to related distributions over next states, and (3) the labels of related states coincide.
When $\Relation$ is a PSR from $\mdp_2$ to $\mdp_1$, we say that \emph{MDP $\mdp_1$ probabilistically simulates MDP $\mdp_2$}.
We denote a PSR from $\mdp_2$ to $\mdp_1$ by $\mdp_2 \preceq \mdp_1$ (loosely speaking, all behaviour of $\mdp_2$ is \emph{contained} in that of $\mdp_1$).

In synthesis problems, $\mdp_2$ is often a (finite-state) \emph{abstraction} of $\mdp_1$.
The next result from~\cite{DBLP:journals/siamco/HaesaertSA17} enables the synthesis of a policy for $\mdp_1$ based on a policy for this abstraction $\mdp_2$.

\begin{theorem}
    \label{thm:PSR_synthesis}
    If $\mdp_2 \preceq \mdp_1$, then for every policy $\policy_2$, there exists a policy $\policy_1$ such that, for all events $\varphi \subset 2^\Labels \times 2^\Labels \times \cdots$,
    \begin{equation}   
        \label{eq:PSR_synthesis}
        \Prob_{\mdp_1}^{\policy_1}\left( \{\labelfunc_1({x_1}_k)\}_{k \in \NN} \in \varphi \right) =
        \Prob_{\mdp_2}^{\policy_2}\left( \{\labelfunc_2({x_2}_k)\}_{k \in \NN} \in \varphi \right).
    \end{equation}
\end{theorem}

The proof, for which we refer to~\cite{DBLP:journals/siamco/HaesaertSA17}, uses that both MDPs induce equal distributions over labelling trajectories.
Intuitively, a policy $\policy_1$ for which \cref{eq:PSR_synthesis} holds is one that preserves the $2^\text{nd}$ PSR condition in \cref{def:PSR}.
Due to space restrictions, we only present this policy explicitly for the MDPs with set-valued dynamics, which we present next. %

\section{Continuous Robust MDPs}
\label{sec:RMDPs}
While the MDP in \cref{def:MDP} defines a very common class of stochastic models, this model definition fundamentally requires the stochastic kernel $\transfunc(\cdot \mid x,u)$ to be known precisely.
This requirement is often unrealistic, especially when the dynamics are estimated from data or subject to set-bounded disturbances, as illustrated by the following example.

\begin{example}
    \label{example2}
    Consider again the Dubins vehicle from \cref{example}.
    Suppose that the parameters $\alpha$, $\beta$ are estimated from (a limited amount of) data and are, therefore, only known up to a given interval, i.e., $\alpha \in [\munderbar{\alpha}, \bar{\alpha}]$, $\beta \in [\munderbar{\beta}, \bar{\beta}]$.
    As a result, the dynamics in \cref{eq:Dubins} have no well-defined stochastic kernel $\transfunc$, so the system cannot be modelled as an MDP.
\end{example}

Motivated by this example, we study a type of MDP with \emph{sets of stochastic kernels}.
Such models are better known as robust MDPs (RMDPs) and have been studied extensively with finite state/action spaces~\cite{DBLP:journals/ior/NilimG05,DBLP:journals/mor/WiesemannKR13}.
Here, we study a variant of RMDPs with continuous state and action spaces.

\begin{definition}[RMDP]
    \label{def:RMDP}
    A (continuous) robust MDP (RMDP) is a tuple $\rmdp = \RMDP$, where
    \begin{itemize}
        \item $\States$, $\initState$, $\Actions$, $\Labels$, and $\labelfunc$ are defined as in \cref{def:MDP},    
        \item $\Disturb$ is a Polish space, called the \emph{disturbance space}, and
        \item $\transfuncR$ is a \emph{stochastic kernel} that assigns to each $x \in \States$, $u \in \Actions$, and $v \in \Disturb$ a probability measure $\transfuncR(\cdot \mid x,u,v)$ over $(\States,\Borel{\States})$,
    \end{itemize}
\end{definition}

The stochastic kernel of an RMDP is, compared to the MDP in \cref{def:MDP}, also conditioned on the disturbance $v \in \Disturb$.
Thus, an RMDP can be interpreted as a 2-player stochastic game, where player 1 chooses an action $u \in \Actions$ and player 2 chooses a disturbance $v \in \Disturb$, which together fix a distribution over next states given by the stochastic kernel $\transfuncR(\cdot \mid x,u,v)$.

\begin{example}
    \label{example3}
    Consider again the Dubins vehicle with uncertain coefficients $\alpha$ and $\beta$ from \cref{example2}. 
    This system can be modelled as an RMDP, where the disturbance space is defined as $\Disturb = [\munderbar{\alpha}, \bar{\alpha}] \times [\munderbar{\beta}, \bar{\beta}]$.
\end{example}

\textbf{Adversary and policy.}
The disturbances in an RMDP are chosen by a (Markov) \emph{adversary} (or \emph{policy of nature}~\cite{DBLP:journals/ior/NilimG05}):

\begin{definition}[Adversary]
    \label{def:adversary}
    A (Markov) adversary $\adversary$ for an RMDP $\rmdp = \RMDP$ is a sequence $\adversary = (\adversary_0, \adversary_1, \ldots)$, where each $\adversary_k$ is a universally measurable map defined as $\adversary_k \colon \States \to \distr{\Disturb}$.
\end{definition}

The definition of a Markov policy (\cref{def:policy}) carries over to RMDPs immediately.
Furthermore, observe that an MDP is a special case of an RMDP with a singleton set $\Disturb$.

\begin{remark}
    The Markovianity of the adversary in \cref{def:adversary} means that the choice of the disturbance $v \in \Disturb$ is independent between the time steps.
    For the Dubins vehicle example, this (conservatively) implies that the adversary can select different parameter values at each step.
    Modelling fixed but unknown parameter values leads to a partially observable model, which drastically increases the complexity of solution methods.
\end{remark}

\textbf{Execution.}
Executions and sample paths for an RMDP are defined by fixing both a policy and an adversary.
That is, an RMDP execution $\{ x_k \}_{k \in \NN}$ is a stochastic process defined on the probability space $(\Omega, \Borel{\Omega}, \Prob_{\rmdp}^{\policy,\adversary})$ with the sample space $\Omega = \States \times \States \times \cdots$, the Borel $\sigma$-algebra~$\Borel{\Omega}$ over $\Omega$, and the (uniquely defined) probability measure $\Prob_{\rmdp}^{\policy,\adversary} \colon \Borel{\Omega} \to [0,1]$.
A sample path is an infinite sequence $\pi = (x_0, x_1, \ldots) \in \Omega$ of states, such that $x_{k+1} \in \support{\transfuncR(\cdot \mid x_k, \policy_k(x_k), \adversary_k(x_k))}$. 
As for MDPs, we use the probability measure $\Prob_{\rmdp}^{\policy,\adversary}$ to reason about the probability that the RMDP satisfies a given specification or control task.

\section{Probabilistic Alternating Simulations}
\label{sec:PASR}
Recall that the PSR from \cref{def:PSR} asserts that, for all related states $(x_1, x_2) \in \Relation$ and for all inputs $u_2 \in \Actions_2$ for MDP $\mdp_2$, there exists an input $u_1 \in \Actions_1$ for MDP $\mdp_1$ such that the resulting kernels $\transfunc_1$ and $\transfunc_2$ are related by the lifted relation $\liftRelation$.
It is apparent that such a PSR is not suited to relate two RMDPs $\rmdp_1$ and $\rmdp_2$, because it does not account for the disturbances $v_1 \in \Disturb_1$ and $v_2 \in \Disturb_2$.
Hence, in this section, we extend the PSR with a condition over the disturbances, leading to a so-called \emph{alternating} notion of simulation~\cite{DBLP:conf/concur/AlurHKV98}.

\subsection{Probabilistic alternating simulation relations}
In an alternating simulation, the matching of related states involves \emph{two} layers of quantification: (1) over the actions $u_2$ and $u_1$, and (2) over the disturbances $v_1$ and $v_2$.
As a key contribution, we extend the PSR from \cref{def:PSR} to RMDPs, by adding this alternation over the disturbances.
This definition is, again, based on lifting a relation $\Relation$ between states, to a relation $\liftRelation$ over distributions (see \cref{def:lifting}).
We first provide the formal definition and discuss its intuition thereafter.

\begin{definition}[Prob. alternating simulation]
    \label{def:PASR}
    Consider two RMDPs $\rmdp_i = \RMDPi{i}$, $i=1,2$ with the same set of labels $\Labels$.
    A single-valued binary relation $\Relation \subseteq \States_1 \times \States_2$ is a \emph{probabilistic alternating simulation relation (PASR)} from $\rmdp_2$ to $\rmdp_1$ if:
    \begin{enumerate}
        \item %
        for the initial distributions, we have $(\initState_1, \initState_2) \in \liftRelation$;
        \item %
        for all $(x_1, x_2) \in \Relation$, we have
        \begin{align}
            &\forall u_2 \in \Actions_2, \, \exists u_1 \in \Actions_1, \, \forall v_1 \in \Disturb_1, \, \exists v_2 \in \Disturb_2
            \label{eq:PASR_next}
            \\
            &\!\!\!\!\!\text{such that}\,\big( \transfuncR_1(\cdot \mid x_1, u_1, v_1), \transfuncR_2(\cdot \mid x_2, u_2, v_2) \big) \in \liftRelation;
            \nonumber
        \end{align}
        \item %
        for all $(x_1, x_2) \in \Relation$, we have $\labelfunc_1(x_1) = \labelfunc_2(x_2)$.
    \end{enumerate}
\end{definition}

Like we write $\mdp_2 \preceq \mdp_1$ to denote a PSR, we write $\rmdp_2 \preceq_\text{alt} \rmdp_1$ to denote a PASR from RMDP $\rmdp_2$ to RMDP $\rmdp_1$.

\subsection{Game interpretation}
Intuitively, condition (2) in \cref{def:PASR} can be interpreted as a game between a \emph{protagonist} and an \emph{antagonist} (which are, importantly, different from the policies $\policy$ and the adversaries $\adversary$ of the RMDPs)~\cite{DBLP:conf/concur/AlurHKV98}.
The antagonist controls the $\forall$-quantifiers, whereas the protagonist controls the $\exists$-quantifiers, i.e.,
\begin{enumerate}
    \item the antagonist chooses an action $u_2 \in \Actions_2$ in $\rmdp_2$;
    \item the protagonist chooses an action $u_1 \in \Actions_1$ in $\rmdp_1$;
    \item the antagonist chooses a disturbance $v_1 \in \Disturb_1$ in $\rmdp_1$;
    \item the protagonist chooses a disturbance $v_2 \in \Disturb_2$ in $\rmdp_2$.
\end{enumerate}
Condition (2) in \cref{def:PASR} requires that, for all $(x_1, x_2) \in \Relation$, the protagonist can choose $u_1$ and $v_2$ such that, no matter what $u_2$ and $v_1$ the antagonist chose, the stochastic kernels $\transfuncR_1$ and $\transfuncR_2$ are related by the lifted relation $\liftRelation$.
This crucial fact will form the basis for policy synthesis with PASRs.

\begin{example}
    \label{example:PASR}
    As a simple example of a PASR, consider the $1$-step RMDPs $\rmdp_1$ and $\rmdp_2$ in \cref{fig:PASR}, where the colors indicate related states.
    For simplicity, suppose $\Actions_i$ and $\Disturb_i$ are all discrete, and that for all $u_i \in \Actions_i$ and $v_i \in \Disturb_i$, the kernels are Dirac distributions.
    We claim that $\rmdp_2 \preceq_\text{alt} \rmdp_1$, i.e., the relation induced by the colouring in \cref{fig:PASR} is a PASR from $\rmdp_2$ to $\rmdp_1$.
    To see why, we can unfold all cases of the game interpretation for condition (2) of \cref{def:PASR}: 
    \begin{itemize}
        \item If the antagonist chooses $u_2$ in $\rmdp_2$, then the protagonist chooses $u_1$ in $\rmdp_1$. Then, if
            (a) the antagonist chooses $v_1$ in $\rmdp_1$, then the protagonist chooses $v'_2$ in $\rmdp_2$, whereas if
            (b) the antagonist chooses $v'_1$ in $\rmdp_1$, then the protagonist chooses $v_2$ in $\rmdp_2$.
        \item If the antagonist chooses $u'_2$ in $\rmdp_2$, then the protagonist also chooses $u_1$ in $\rmdp_1$. Then, if
            (a) the antagonist chooses $v_1$ in $\rmdp_1$, then the protagonist chooses $v_2$ in $\rmdp_2$, whereas if
            (b) the antagonist chooses $v'_1$ in $\rmdp_1$, then the protagonist chooses $v'_2$ in $\rmdp_2$.
    \end{itemize}
    Observe that all cases lead to related next states (i.e., states with the same colour in \cref{fig:PASR}), thus preserving the PASR.
\end{example}

\subsection{Policy refinement}

The existence of a PASR between two RMDPs can (like a PSR between two MDPs) be used to synthesise Markov policies.
Fix RMDPs $\rmdp_1$ and $\rmdp_2$, and let $\Relation \subseteq \States_1 \times \States_2$ be a PASR from $\rmdp_2$ to $\rmdp_1$, i.e., $\rmdp_2 \preceq_\text{alt} \rmdp_1$.
An interface function \emph{refines} a policy $\policy_2$ for $\rmdp_2$ into a policy $\policy_1$ for $\rmdp_1$ such that the PASR is preserved.

\begin{figure}[t!]
\centering
\includegraphics[scale=0.88]{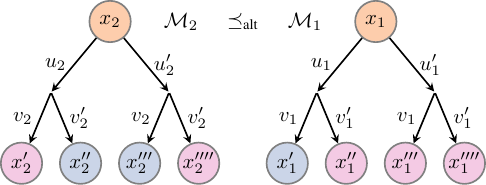}%
\caption{Visualisation for a single step of condition (2) in \cref{def:PASR}, for a PASR from RMDP $\rmdp_2$ to $\rmdp_1$, i.e., $\rmdp_2 \preceq_\text{alt} \rmdp_1$.}
\label{fig:PASR}
\end{figure}

\begin{definition}[Interface function]
    \label{def:interface}
    An \emph{interface (function)} $I \colon \States_1 \times \States_2 \times \Actions_2 \to 2^{\Actions_1}$ from $\rmdp_2$ to $\rmdp_1$ is a set-valued map defined for all $(x_1, x_2) \in \Relation$ and $u_2 \in \Actions_2$ as
    \begin{align*}
    I(x_1, x_2, u_2) {}&{} = \Big\{ 
        u_1 \in \Actions_1 : \,\,\,
        \forall v_1 \in \Disturb_1, \, \exists v_2 \in \Disturb_2, \, 
        \\
        & \left( \transfuncR_1(\cdot \mid x_1, u_1, v_1), \transfuncR_2(\cdot \mid x_2, u_2, v_2) \right) \in \liftRelation
    \Big\}.
    \end{align*}
\end{definition}

\begin{lemma}[Nonemptyiness]
    \label{lemma:nonempty_interface}
    $\rmdp_2 \preceq_\text{alt} \rmdp_1$ implies that $I(x_1, x_2, u_2)$ is nonempty for all $x_1 \in \States_1, x_2 \in \States_2, u_2 \in \Actions_u$.
\end{lemma}

\begin{proof}
    A PASR $\Relation$ is single-valued by definition, so for all $x_1 \in \States_1$, there exists an $x_2 \in \States_2$ such that $(x_1,x_2) \in \Relation$.
    By \cref{def:PASR}, for all $(x_1,x_2) \in \Relation$ and all $u_2 \in \Actions_2$, there exists an action $u_1 \in \Actions_1$ such that 
    \begin{align*}
    \big( \transfuncR_1(\cdot \mid x_1, u_1, v_1), {}&{} \transfuncR_2(\cdot \mid x_2, u_2, v_2) \big) \in \liftRelation, 
    \\
    &\forall v_1 \in \Disturb_1, \, \exists v_2 \in \Disturb_2,
    \end{align*}
    which equals the definition of the interface, so $I(x_1, x_2, u_2) \neq \emptyset$ for all $x_1 \in \States_1$, $x_2 \in \States_2$, and $u_2 \in \Actions_2$.
\end{proof}

Towards the main result, we present the following lemma, which states that, under a PASR $\rmdp_2 \preceq_\text{alt} \rmdp_1$ and an interface function, a pair of related states $(x_1, x_2) \in \Relation$ leads to equal distributions over labels $2^\Labels$ in the next states.
For this lemma, let $\Prob^\rmdp_{x,u,v}(x' \in A) = \int_A \transfuncR(dy \mid x, u, v)$ be the probability that the next state $x'$ is contained in $A \in \Borel{\States}$ when the current state is $x$, and action $u$ and disturbance $v$ are executed.

\begin{lemma}
    \label{lemma:observations}
    Let $\rmdp_1$ and $\rmdp_2$ be two RMDPs such that $\rmdp_2 \preceq_\text{alt} \rmdp_1$.
    Fix $(x_1, x_2) \in \Relation$, $u_2 \in \Actions_2$, and $u_1 \in I(x_1, x_2, u_2)$.
    Then, for all $v_1 \in \Disturb_1$, there exists $v_2 \in \Disturb_2$ such that for all subsets of labels $L \in 2^\Labels$, it holds that
    \begin{equation}
        \label{eq:lemma_observation}
        \Prob^{\rmdp_1}_{x_1,u_1,v_1}(h_1(x_1') = L) = 
        \Prob^{\rmdp_2}_{x_2,u_2,v_2}(h_2(x_2') = L).
    \end{equation}
\end{lemma}

\begin{proof}
    By \cref{def:interface}, restricting $u_1$ to the interface function $I(x_1, x_2, u_2)$ implies that condition (2) of the PASR in \cref{def:PASR} is satisfied, i.e., $\forall v_1 \in \Disturb_1, \, \exists v_2 \in \Disturb_2$ such that
    \begin{equation}
        \label{eq:lemma_proof1}
        \big( \transfuncR_1(\cdot \mid x_1, u_1, v_1), \transfuncR_2(\cdot \mid x_2, u_2, v_2) \big) \in \liftRelation.
    \end{equation}
    By \cref{def:lifting} of the lifted relation $\liftRelation$, \cref{eq:lemma_proof1} implies that for all $X_1 \in \Borel{\States_1}$, it holds that
    $
    \Prob^{\rmdp_1}_{x_1,u_1,v_1}(x'_1 \in X_1) = \Prob^{\rmdp_2}_{x_2,u_2,v_2}(x'_2 \in \Relation(X_1)).
    $
    Conversely, for all $X_2 \in \Borel{\States_2}$,
    $
    \Prob^{\rmdp_2}_{x_2,u_2,v_2}(x'_2 \in X_2) = \Prob^{\rmdp_1}_{x_1,u_1,v_1}(x'_1 \in \Relinv{X_2}).
    $
    Finally, since the labelling functions $h_1$ and $h_2$ are Borel measurable, we arrive at \cref{eq:lemma_observation} and thus conclude the proof.
\end{proof}

The following theorem is the main result of this paper and shows that a PASR $\rmdp_2 \preceq_\text{alt} \rmdp_1$ allows to \emph{refine} any policy $\policy_2$ for RMDP $\rmdp_2$ (i.e., the abstraction) to a policy $\policy_1$ for RMDP $\rmdp_1$ (i.e., the concrete system).
This refined policy has \emph{at least} the same probability of satisfying any given behavioural specification.

\begin{theorem}[Policy refinement]
    \label{thm:PASR_synthesis}
    Let $\rmdp_1$ and $\rmdp_2$ be two RMDPs.
    If $\rmdp_2 \preceq_\text{alt} \rmdp_1$, then for all policies $\policy_2$ and all events $\varphi \subset 2^\Labels \times 2^\Labels \times \cdots$, it holds that
    \begin{equation}
    \begin{split}
        \label{eq:PASR_synthesis}
        &\min_{\adversary_1} \Prob_{\rmdp_1}^{\policy_1, \adversary_1}\left( \{\labelfunc_1({x_1}_k)\}_{k \in \NN} \in \varphi \right) \geq \\
        &\qquad\qquad\qquad\min_{\adversary_2} \,\Prob_{\rmdp_2}^{\policy_2, \adversary_2}\left( \{\labelfunc_2({x_2}_k)\}_{k \in \NN} \in \varphi \right),
    \end{split}
    \end{equation}
    where the policy $\policy_1$ is defined for all $k \in \NN$ and $x_1 \in \States$ as ${\policy_1}_k(x_1) \in I(x_1, x_2, {\policy_2}_k(x_2))$, with 
    $x_2 \in \Relation(x_1)$.
\end{theorem}

\begin{proof}
    We will prove the theorem by showing that, for every $\tilde\adversary_1$ in $\rmdp_1$, there exists a $\tilde\adversary_2$ in $\rmdp_2$ such that
    \begin{equation}
        \Prob_{\rmdp_1}^{\policy_1, \tilde\adversary_1}\left( \{\labelfunc_1(x_k)\}_{k \in \NN} \in \varphi \right)
        =
        \Prob_{\rmdp_2}^{\policy_2, \tilde\adversary_2}\left( \{\labelfunc_2(x_k)\}_{k \in \NN} \in \varphi \right).
        \label{eq:thm_proof1}
    \end{equation}
    If for all $\tilde\adversary_1$, there exists $\tilde\adversary_2$ such that \cref{eq:thm_proof1} holds, then for $\adversary^\star_1 \in \argmin_{\adversary_1} \Prob_{\rmdp_1}^{\policy_1, \adversary_1}\left( \{\labelfunc_1(x_k)\}_{k \in \NN} \in \varphi \right)$, there exists $\tilde\adversary_2$ s.t.
    \begin{equation*}
        \Prob_{\rmdp_1}^{\policy_1, \tilde\adversary_1^\star}\left( \{\labelfunc_1(x_k)\}_{k \in \NN} \in \varphi \right)
        =
        \Prob_{\rmdp_2}^{\policy_2, \tilde\adversary_2}\left( \{\labelfunc_2(x_k)\}_{k \in \NN} \in \varphi \right).
        \label{eq:thm_proof12}
    \end{equation*}
    Thus, $\min_{\adversary_1} \Prob_{\rmdp_1}^{\policy_1, \adversary_1}\left( \{\labelfunc_1(x_k)\}_{k \in \NN} \!\in\!\varphi \right)$ cannot be smaller than $\min_{\adversary_2} \Prob_{\rmdp_2}^{\policy_2, \adversary_2}\left( \{\labelfunc_2(x_k)\}_{k \in \NN} \!\in\! \varphi \right)$, and thus, \cref{eq:PASR_synthesis} follows.

    What remains is to show that \cref{eq:thm_proof1} holds.
    In fact, \cref{eq:thm_proof1} follows from \cref{lemma:observations}: Given related states $(x_1, x_2)$, the distributions over the next observations coincide.
    Moreover, the next states remain related, so subsequent distributions over observations also coincide.
    Thus, \cref{thm:PASR_synthesis} follows.
\end{proof}

\begin{remark}
    If the interface function in \cref{def:interface} is given in explicit form, then \cref{thm:PASR_synthesis} reduces to a look-up step and is thus tractable.
    Yet, computing this interface can be challenging, especially for general nonlinear dynamics.
\end{remark}

\subsection{Discussion}
In this paper, we defined specifications for (R)MDPs as \emph{sets of labelling trajectories}, that is, $\varphi \subset 2^\Labels \times 2^\Labels \times \cdots$.
A common example of such a specification is the (infinite-horizon) \emph{reach-avoid} specification, which is satisfied if the system reaches the goal states $X_G \subset \States$ while avoiding the unsafe states $X_U \subset \States$.
Let $\Labels = \{ \mathsf{G}, \mathsf{U} \}$ and define the labelling function $\labelfunc \colon \States \to 2^\Labels$ for all $x \in \States$ as
$
   x \in X_G \iff \mathsf{G} \in \labelfunc(x),
    \,\text{and}\,
   x \in X_U \iff \mathsf{U} \in \labelfunc(x).
$
The corresponding reach-avoid specification $\varphi_\text{rwa} \subset 2^\Labels \times 2^\Labels \times \cdots$ is defined as
\begin{align*}\label{equ:RwA}
    \varphi_\text{rwa} \coloneqq \big\{ 
   (\labelfunc(x_0), \labelfunc(x_1), \ldots) : {}&{}\exists k \in \NN, \, \mathsf{G} \in \labelfunc(x_k) \, \wedge
    \\
    &\forall k'\leq k, \, \mathsf{U} \notin \labelfunc(x_{k'})
    \big\}.
\end{align*}
In practice, it is often convenient to express specifications in temporal logic, such as LTL and PCTL; however, we omit further details and refer to~\cite{BaierKatoen08} for a textbook introduction.

Several papers construct RMDP or IMDP abstractions of stochastic dynamical systems~\cite{DBLP:journals/jair/BadingsRAPPSJ23,DBLP:journals/tac/LahijanianAB15,DBLP:conf/hybrid/CauchiLLAKC19,DBLP:conf/l4dc/GraciaBLL24,DBLP:journals/corr/abs-2404-08344}.
Often, the correctness of such approaches implicitly relies on establishing a PASR from the abstraction to the concrete system.
For example,~\cite{DBLP:conf/aaai/BadingsRA023} studies abstraction-based control of stochastic dynamical systems with set-bounded uncertain parameters.
Their setting is a special case of ours, where the concrete model is a continuous-state/action RMDP as per \cref{def:RMDP}, and where the abstract model is a finite-state interval MDP (IMDP), which is an RMDP where the transition probabilities are defined as intervals.
Our probabilistic alternating simulation relation makes the analysis of~\cite{DBLP:conf/aaai/BadingsRA023} more explicit and thus contributes to a better formalisation of abstraction-based controller synthesis techniques.
Finally, PASR can also be used for state space reduction in finite RMDPs.

\section{Numerical experiment}
\label{sec:experiments}
We demonstrate the applicability of our techniques to synthesise a finite-state interval MDP (IMDP) abstraction for the 4D-state Dubins vehicle with uncertain parameters from \cref{example}.
The experiments ran on an Apple MacBook with an M4 Pro chip and 24GB of RAM.
Our Python code is available via \url{https://github.com/LAVA-LAB/dynabs-jax} and uses JAX %
for just-in-time (JIT) compilation.

\textbf{Dynamics.}
We consider the dynamics from \cref{example} with a time discretisation of $\delta = 0.5$.
We set the true parameters to $\alpha^\star = 0.85$, $\beta^\star = 0.85$.
The goal is to synthesise a policy that maximises the probability to satisfy the reach-avoid specification in \cref{fig:Dubins} (goal states $X_G$ in green; unsafe states $X_U$ in red; only position variables $(x,y)$ shown).
We constrain the vehicle's speed to $V_k \in [-3, 3]$, the steering input to $u_k \in [-0.5\pi, 0.5\pi]$, and the acceleration input to $u'_k \in [-5, 5]$.

\textbf{Abstraction.}
We follow a relatively standard approach to constructing the IMDP abstraction, similar to, e.g.,~\cite{DBLP:conf/hybrid/CauchiLLAKC19,DBLP:conf/hybrid/MathiesenHL25,DBLP:journals/tac/LahijanianAB15}.
We refer to the concrete model as $\rmdp_1$ and to the IMDP as $\rmdp_2$.
We partition the state space into $40 \times 20 \times 20 \times 20 = 320\,000$ states and uniformly grid the input space into $7 \times 7$ actions. %
As in \cref{example3}, we model the uncertain parameters $\alpha$ and $\beta$ using the IMDP's disturbances $\Disturb_2$ .
We compute the probability intervals of the IMDP by adapting the approach from~\cite{DBLP:journals/tac/LahijanianAB15} to uncertain parameters.
Intuitively, the probability $\transfuncR_2(x'_2 \mid x_2, u_2, v_2)$ of reaching a state $x'_2 \in \States_2$ by executing the action $u_2 \in \Actions_2$ in state $x_2 \in \States$ is obtained by integrating the kernel $\transfuncR_1$ of the concrete model over the associated concrete states $\Relinv{x'_2} \subset \States_1$ and taking the min/max over the disturbances $v_2 \in \Disturb$ (representing all possible values $\alpha$ and $\beta$).
As the process noise is additive and Gaussian, we can efficiently compute these probability intervals.
We use robust value iteration implemented in the model checker Storm~\cite{DBLP:conf/cav/DehnertJK017}, to compute an optimal policy on the IMDP abstraction.

\begin{figure}[t]
\begin{subfigure}[b]{0.495\linewidth}
    \centering
    \includegraphics[width=\linewidth]{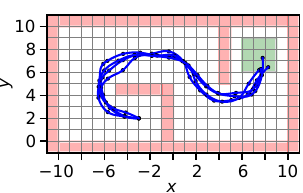}
    \caption{Case (1): known $\alpha$ and $\beta$.}
    \label{fig:Dubins_case0}
\end{subfigure}\hfill
\begin{subfigure}[b]{0.495\linewidth}
    \centering
    \includegraphics[width=\linewidth]{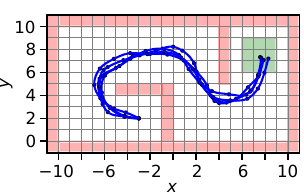}
    \caption{Case (2): uncertain $\alpha$ and $\beta$.}
    \label{fig:Dubins_case2}
\end{subfigure}
    \caption{Simulations of the 4D-state Dubins vehicle under the policies synthesised using \cref{thm:PASR_synthesis}. Even though the parameter uncertainty increases the number of transitions, the performance of the resulting policy is practically unaffected.}
    \label{fig:Dubins}
\end{figure}

\textbf{Cases.}
We compare two cases: (1) the parameters $\alpha$ and $\beta$ are precisely known, and (2) the parameters are only known up to $\alpha \in [0.8, 0.9]$ and $\beta \in [0.8, 0.9]$.
For both cases, we construct the abstract IMDP $\rmdp_2$ described above and use \cref{thm:PASR_synthesis} to refine an (optimal) IMDP policy $\policy_2$ into a policy $\policy_1$ for the Dubins vehicle $\rmdp_1$ together with a lower bound on the probability of satisfying the reach-avoid specification (which is obtained as the right-hand side of \cref{eq:PASR_synthesis}).

\textbf{Results.}
Without parameter uncertainty, generating the IMDP takes around $\SI{9}{\minute}$, and computing an optimal IMDP policy $\SI{5}{\minute}$.
With uncertainty, generating the IMDP and computing an optimal IMDP policy takes around $16$ and $\SI{8}{\minute}$, respectively.
For both cases, the IMDPs have $320\,000$ states, but adding parameter uncertainty increases the number of transitions (i.e., the number of edges in the underlying graph of the IMDP) from 205 million to 354 million.
Indeed, the uncertain parameters lead to additional transitions between states that must be modelled in the IMDP.
We also tested a third case with even more uncertainty (where $\alpha \in [0.7, 1.0]$ and $\beta \in [0.7, 1.0]$); however, this led to a vacuous IMDP abstraction with too much conservatism in the transitions.

Without parameter uncertainty, the bound on the satisfaction probability obtained using \cref{thm:PASR_synthesis} is $\rho^\star = \min_{\adversary_2} \,\Prob_{\rmdp_2}^{\policy_2, \adversary_2}\left( \{\labelfunc_2({x_2}_k)\}_{k \in \NN} \in \varphi \right) = 0.996$.
With parameter uncertainty, we obtain a (negligibly lower) bound of $\rho^\star = 0.995$.
To validate these bounds, we run $10\,000$ simulations of the concrete model under the synthesised policies and the true parameters $\alpha^\star$ and $\beta^\star$.
Four state trajectories under the policies for both cases are shown in \cref{fig:Dubins}.
Interestingly, the trajectories for both cases are almost identical.
We believe this is because the parameter uncertainty only directly affects the speed ($V_k$) and steering angle ($\theta_k$) variable, but \cref{fig:Dubins} only shows the position ($x_k$ and $y_k$).
For both cases, all simulated trajectories satisfy the reach-avoid specification, showing that the theoretical bounds are indeed achieved in practice.

\section{Conclusion}
We presented a notion of probabilistic alternating simulation between robust MDPs (RMDPs) with continuous state and action spaces.
Such continuous RMDPs are useful to model systems with both stochastic and nondeterministic (i.e., set-valued) dynamics.
We showed how to use probabilistic alternating simulation relations (PASR) to synthesise policies that provably satisfy complex specifications.
We demonstrated the applicability of our techniques on a reach-avoid problem for a 4D-state Dubins vehicle with uncertain parameters.

In the future, we aim to apply our techniques for model order reduction by using a PASR to relate two continuous RMDPs.
We also plan to study approximate versions of PASR to enable solving more challenging control problems, similar to the approximate probabilistic simulation developed by, e.g.,~\cite{DBLP:journals/siamco/HaesaertSA17}.
Finally, we wish to more explicitly connect our results to the relations for continuous stochastic games in~\cite{DBLP:journals/automatica/ZhongLZC23}.

\bibliographystyle{ieeetr}
\bibliography{references.bib}

\begin{thebibliography}{10}

\bibitem{DBLP:books/wi/Puterman94}
M.~L. Puterman, {\em Markov Decision Processes: Discrete Stochastic Dynamic Programming}.
\newblock John Wiley \& Sons, 2014.

\bibitem{BaierKatoen08}
C.~Baier and J.~Katoen, {\em Principles of model checking}.
\newblock {MIT} Press, 2008.

\bibitem{DBLP:conf/cav/AbateGR25}
A.~Abate, M.~Giacobbe, and D.~Roy, ``Quantitative supermartingale certificates,'' in {\em {CAV} {(2)}}, vol.~15932 of {\em LNCS}, pp.~3--28, Springer, 2025.

\bibitem{DBLP:journals/tac/PrajnaJP07}
S.~Prajna, A.~Jadbabaie, and G.~J. Pappas, ``A framework for worst-case and stochastic safety verification using barrier certificates,'' {\em {IEEE} Trans. Autom. Control.}, vol.~52, no.~8, pp.~1415--1428, 2007.

\bibitem{DBLP:journals/automatica/AbatePLS08}
A.~Abate, M.~Prandini, J.~Lygeros, and S.~Sastry, ``Probabilistic reachability and safety for controlled discrete time stochastic hybrid systems,'' {\em Autom.}, vol.~44, no.~11, pp.~2724--2734, 2008.

\bibitem{DBLP:journals/tac/ZamaniEMAL14}
M.~Zamani, P.~M. Esfahani, R.~Majumdar, A.~Abate, and J.~Lygeros, ``Symbolic control of stochastic systems via approximately bisimilar finite abstractions,'' {\em {IEEE} Trans. Autom. Control.}, vol.~59, no.~12, pp.~3135--3150, 2014.

\bibitem{DBLP:journals/tac/LahijanianAB15}
M.~Lahijanian, S.~B. Andersson, and C.~Belta, ``Formal verification and synthesis for discrete-time stochastic systems,'' {\em {IEEE} Trans. Autom. Control.}, vol.~60, no.~8, pp.~2031--2045, 2015.

\bibitem{DBLP:journals/jair/BadingsRAPPSJ23}
T.~S. Badings, L.~Romao, A.~Abate, D.~Parker, H.~A. Poonawala, M.~Stoelinga, and N.~Jansen, ``Robust control for dynamical systems with non-gaussian noise via formal abstractions,'' {\em J. Artif. Intell. Res.}, vol.~76, pp.~341--391, 2023.

\bibitem{DBLP:conf/hybrid/MathiesenHL25}
F.~B. Mathiesen, S.~Haesaert, and L.~Laurenti, ``Scalable control synthesis for stochastic systems via structural {IMDP} abstractions,'' in {\em {HSCC}}, pp.~14:1--14:12, {ACM}, 2025.

\bibitem{DBLP:conf/l4dc/GraciaBLL24}
I.~Gracia, D.~Boskos, L.~Laurenti, and M.~Lahijanian, ``Data-driven strategy synthesis for stochastic systems with unknown nonlinear disturbances,'' in {\em {L4DC}}, vol.~242 of {\em PMLR}, pp.~1633--1645, 2024.

\bibitem{pmlr-v283-nazeri25a}
M.~Nazeri, T.~Badings, S.~Soudjani, and A.~Abate, ``Data-driven yet formal policy synthesis for stochastic nonlinear dynamical systems,'' in {\em {L4DC}}, vol.~283 of {\em PMLR}, pp.~1550--1564, 2025.

\bibitem{DBLP:journals/csysl/LavaeiSFZ23}
A.~Lavaei, S.~Soudjani, E.~Frazzoli, and M.~Zamani, ``Constructing {MDP} abstractions using data with formal guarantees,'' {\em {IEEE} Control. Syst. Lett.}, vol.~7, pp.~460--465, 2023.

\bibitem{DBLP:journals/automatica/LavaeiSAZ22}
A.~Lavaei, S.~Soudjani, A.~Abate, and M.~Zamani, ``Automated verification and synthesis of stochastic hybrid systems: {A} survey,'' {\em Autom.}, vol.~146, p.~110617, 2022.

\bibitem{DBLP:books/daglib/0032856}
P.~Tabuada, {\em Verification and Control of Hybrid Systems - {A} Symbolic Approach}.
\newblock Springer, 2009.

\bibitem{DBLP:journals/tac/GirardP07}
A.~Girard and G.~J. Pappas, ``Approximation metrics for discrete and continuous systems,'' {\em {IEEE} Trans. Autom. Control.}, vol.~52, no.~5, pp.~782--798, 2007.

\bibitem{DBLP:journals/tac/ReissigWR17}
G.~Reissig, A.~Weber, and M.~Rungger, ``Feedback refinement relations for the synthesis of symbolic controllers,'' {\em {IEEE} Trans. Autom. Control.}, vol.~62, no.~4, pp.~1781--1796, 2017.

\bibitem{DBLP:conf/hybrid/CalbertMGJ24}
J.~Calbert, S.~M. Mattenet, A.~Girard, and R.~M. Jungers, ``Memoryless concretization relation,'' in {\em {HSCC}}, pp.~14:1--14:9, {ACM}, 2024.

\bibitem{DBLP:journals/siamco/HaesaertSA17}
S.~Haesaert, S.~E.~Z. Soudjani, and A.~Abate, ``Verification of general markov decision processes by approximate similarity relations and policy refinement,'' {\em {SIAM} J. Control. Optim.}, vol.~55, no.~4, pp.~2333--2367, 2017.

\bibitem{DBLP:journals/ior/NilimG05}
A.~Nilim and L.~E. Ghaoui, ``Robust control of markov decision processes with uncertain transition matrices,'' {\em Oper. Res.}, vol.~53, no.~5, pp.~780--798, 2005.

\bibitem{DBLP:journals/mor/WiesemannKR13}
W.~Wiesemann, D.~Kuhn, and B.~Rustem, ``Robust markov decision processes,'' {\em Math. Oper. Res.}, vol.~38, no.~1, pp.~153--183, 2013.

\bibitem{DBLP:conf/hybrid/Delimpaltadakis23}
G.~Delimpaltadakis, M.~Lahijanian, M.~{Mazo Jr.}, and L.~Laurenti, ``Interval markov decision processes with continuous action-spaces,'' in {\em {HSCC}}, pp.~12:1--12:10, {ACM}, 2023.

\bibitem{DBLP:conf/aistats/PanagantiK22}
K.~Panaganti and D.~M. Kalathil, ``Sample complexity of robust reinforcement learning with a generative model,'' in {\em {AISTATS}}, vol.~151 of {\em PMLR}, pp.~9582--9602, {PMLR}, 2022.

\bibitem{DBLP:conf/concur/AlurHKV98}
R.~Alur, T.~A. Henzinger, O.~Kupferman, and M.~Y. Vardi, ``Alternating refinement relations,'' in {\em {CONCUR}}, vol.~1466 of {\em LNCS}, pp.~163--178, Springer, 1998.

\bibitem{DBLP:conf/sofsem/ZhangP12}
C.~Zhang and J.~Pang, ``An algorithm for probabilistic alternating simulation,'' in {\em {SOFSEM}}, vol.~7147 of {\em LNCS}, pp.~431--442, Springer, 2012.

\bibitem{DBLP:journals/automatica/ZhongLZC23}
B.~Zhong, A.~Lavaei, M.~Zamani, and M.~Caccamo, ``Automata-based controller synthesis for stochastic systems: {A} game framework via approximate probabilistic relations,'' {\em Autom.}, vol.~147, p.~110696, 2023.

\bibitem{Bertsekas.Shreve78}
D.~P. Bertsekas and S.~E. Shreve, {\em Stochastic Optimal Control: The Discrete-time Case}.
\newblock Athena Scientific, 1978.

\bibitem{DBLP:conf/hybrid/CauchiLLAKC19}
N.~Cauchi, L.~Laurenti, M.~Lahijanian, A.~Abate, M.~Kwiatkowska, and L.~Cardelli, ``Efficiency through uncertainty: scalable formal synthesis for stochastic hybrid systems,'' in {\em {HSCC}}, pp.~240--251, {ACM}, 2019.

\bibitem{DBLP:journals/corr/abs-2404-08344}
R.~Coppola, A.~Peruffo, L.~Romao, A.~Abate, and M.~{Mazo Jr.}, ``Data-driven interval {MDP} for robust control synthesis,'' {\em CoRR}, vol.~abs/2404.08344, 2024.

\bibitem{DBLP:conf/aaai/BadingsRA023}
T.~S. Badings, L.~Romao, A.~Abate, and N.~Jansen, ``Probabilities are not enough: Formal controller synthesis for stochastic dynamical models with epistemic uncertainty,'' in {\em {AAAI}}, pp.~14701--14710, {AAAI} Press, 2023.

\bibitem{DBLP:conf/cav/DehnertJK017}
C.~Dehnert, S.~Junges, J.~Katoen, and M.~Volk, ``A storm is coming: {A} modern probabilistic model checker,'' in {\em {CAV} {(2)}}, vol.~10427 of {\em LNCS}, pp.~592--600, Springer, 2017.

\end{thebibliography}

\end{document}